\documentclass[runningheads, letter]{llncs}
\usepackage{graphicx}
\usepackage{extarrows}
\usepackage{amssymb,amscd}
\usepackage[all]{xy}
\usepackage{longtable}
\usepackage[latin1]{inputenc}
\usepackage{tikz,arydshln}
\usetikzlibrary{shapes,arrows}
\usepackage{booktabs}
\usepackage{float}
\usepackage{algorithm}
\usepackage{algpseudocode}
\usepackage{url}

\newcommand{\keywords}[1]{\par\addvspace\baselineskip
\noindent\keywordname\enspace\ignorespaces#1}

\newtheorem{ex}{Example}
\spnewtheorem{contribution}{Contribution}{\bfseries}{\itshape}

\newcommand{\rmv}[1]{}

\newcommand{\Z}{{\mathbb Z}}
\newcommand{\R}{{\mathbb R}}

\newcommand{\C}{{\mathbb C}}
\newcommand{\Q}{{\mathbb Q}}
\newcommand{\K}{{\mathbb K}}
\renewcommand{\L}{{\mathcal L}}

\renewcommand{\b}{{\mathbf b}}
\newcommand{\B}{{\mathbf B}}
\renewcommand{\O}{{\mathcal O}}

\newcommand{\fp}{\mathfrak{p}}

\newcommand{\fa}{\mathfrak{a}}
\begin{document}
\mainmatter
\title
{Principal ideal problem and \\ideal shortest vector over rational primes \\ in power-of-two cyclotomic fields }

\author{Gaohao Cui\inst{1,2,5} and Jianing Li \inst{3,4,5} and Jincheng Zhuang\inst{1,2,5}}
\institute{
School of Cyber Science and Technology, Shandong University, \\ Qingdao 266237, China.
\and Key Laboratory of Cryptologic Technology and Information Security of Ministry of Education, Shandong University, Qingdao, 266237, China
\and Research Center for Mathematics and Interdisciplinary Sciences
Shandong University, Qingdao 266237, China
\and Frontiers Science Center for Nonlinear Expectations, Ministry of Education, Qingdao 266237, China
\and State Key Laboratory of Cryptography and Digital Economy Security, \\ Shandong University, Qingdao, 266237, China}

\maketitle

\begin{abstract}

    The shortest vector problem (SVP) over ideal lattices is closely related to the Ring-LWE problem, which is widely used to build post-quantum cryptosystems. Power-of-two cyclotomic fields are frequently adopted to instantiate Ring-LWE. Pan et al. (EUROCRYPT~2021) explored the SVP over ideal lattices via the decomposition fields and, in particular determined the length of the shortest vector in prime ideals lying over rational primes $p\equiv3,5\pmod{8}$ in power-of-two cyclotomic fields via explicit construction of reduced lattice bases. 

    In this work, we first provide a new method (different from analyzing lattice bases) to analyze the length of the shortest vector in prime ideals in $\mathbb{Z}[\zeta_{2^{n+1}}]$ when $p\equiv3,5\pmod{8}$. Then we precisely characterize the length of the shortest vector in the cases of $p\equiv7,9\pmod{16}$. Furthermore, we derive a new upper bound $\sqrt[4]{2^{2n+1}p}$ for this length, which is tighter than the bound $2^n\sqrt[4]{p}$ obtained from Minkowski's theorem. Our key technique is to investigate whether a generator of a principal ideal can achieve the shortest length after embedding as a vector. If this holds for the ideal, finding the shortest vector in this ideal can be reduced to finding its shortest generator.

\end{abstract}

\section{Introduction}
Lattice-based cryptography has emerged as a cornerstone of post-quantum cryptography, offering a robust alternative to classical public-key systems that are vulnerable to attacks by quantum computers.
The security of these schemes is fundamentally rooted in the computational difficulty of solving hard lattice problems. In particular, several central problems include NTRU problem proposed by Hoffstein, Pipher and Silverman \cite{HoffsteinPS98}, SIS problem proposed by Ajtai \cite{Ajtai96} and the 
learning with errors (LWE) problem introduced by Regev \cite{Regev09}. 

The efficiency of NTRU benefits from its ring structure.  In order to utilize algebraic structures to improve efficiency, Micciancio \cite{MicciancioP12} proposed instantiating SIS instances from rings, Stehl{\'{e}} et al. \cite{StehleSTX09}
and Lyubashevsky et al. \cite{LPR10} proposed instantiating LWE instances from rings. Their hardness is closely related to ideal lattice SVP.
The ideal lattice problem is also used in constructing other schemes, such as fully homomorphic encryption (FHE) schemes \cite{smart2010fully}. 

The ideal lattices derived from power-of-two cyclotomic fields are especially popular due to their desirable structural properties, which facilitate fast polynomial multiplication and efficient key generation. The theoretical security of these schemes is inextricably linked to our ability to understand and characterize the behavior of shortest vectors in these specific lattices. A deeper understanding of these problems not only strengthens our cryptographic foundations but also guides the design of more secure and efficient cryptographic protocols.

\subsection{Related works}
Given the importance of the ideal lattice, there are several approaches to solve the shortest vector problem (SVP) over these structure lattices. 

One line of research is to study the principal ideal problem (PIP), which has several variants including decision version (to decide whether a given ideal is principal), search version (to find a generator given a principal ideal) and optimization version (to find a short generator given a principal ideal).
Biasse et al.~\cite{biasse2017computing} proposed a classical subexponential time heuristic algorithm to find a generator in cyclotomic integer rings.
Biasse and Song~\cite{BS16} designed a polynomial time quantum algorithm to find a generator over arbitrary classes of number fields based on the work of~\cite{EHKS14}.
Cramer, Ducas, Peikert, and Regev~\cite{CDPR16} showed how to find a shorter generator given a random generator using the log-unit lattice in prime-power cyclotomic fields.
Given a cyclotomic number field of prime power conductor with degree $n$,
Cramer, Ducas and Wesolowski~\cite{CDW17} designed a quantum polynomial-time algorithm to solve $2^{\tilde{O}(\sqrt{n})}-$ideal SVP.

Another line of research is to solve the ideal SVP using the subfield attack.
Bernstein proposed logarithm-subfield attack against ideal lattices \cite{B14}. There are several works to attack overstretched NTRU, such as Albrecht, Bai and Ducas \cite{ABD16}, Cheon, Jeong and Lee \cite{CJL16}, and Kirchner and Fouque \cite{KF17}.
    
As a special case of subfield attack, Pan et al. \cite{PXWC21} established a connection between the complexity of solving the ideal SVP of a prime ideal in a number field and the properties of the prime's decomposition group. Using this insight, they were able to explicitly determine the length of the shortest vector for prime ideals corresponding to two specific classes of rational primes $p\equiv \pm 3 \pmod 8$.

Dong et al., building on the work in \cite{dong2022subfield}, proposed a subfield attack on Hermite-SVP in ideal lattices over arbitrary number fields, and generalized the decomposition fields to arbitrary subfields. Porter et al. \cite{porter2023subfield} extended Pan et al.'s work to more classes of ideals. Then, Boudgoust et al. \cite{boudgoust2022some} generalized the work of \cite{PXWC21} and \cite{porter2023subfield} and showed that ideal SVP can be solved efficiently for ideal lattices with a lot of symmetries.

\subsection{Research gap and our results}
   
Pan et al. explicitly determined the length of the shortest vector for prime ideals corresponding to two specific classes of rational primes $p\equiv \pm 3 \pmod 8$. They also proposed the following open problem.
   \begin{quote}
       \textit{It is an interesting problem to study the length of the shortest vectors in other prime ideals.}
   \end{quote}

    For the two specific classes of rational primes $p\equiv \pm 3 \pmod 8$, Pan et al. analyzed the length of the shortest vector for prime ideals in $\Z[\zeta_{2^{n+1}}]$ through reducing it to determining the length of the shortest vector in the lattice generated by prime ideals in $\Z[i]$ or $\Z[\zeta_8]$. However, for classes of rational primes $p\equiv\pm1\pmod{8}$, although we can still reduce the problem to determining the length of the shortest vector in some lower dimension lattice, it is difficult to directly find an explicit relationship between $p$ and the length by analyzing the lattice bases.
    
   In this work, we combine the principal ideal problem approach and the decomposition subfield approach to investigate the length of the shortest vector in prime ideals of $\Z[\zeta_{2^{n+1}}]$ lying over $p\equiv7,9\pmod{16}$. Our contributions are as follows:
   \begin{enumerate}
       \item We consider the problem: given a principal ideal, must the \textit{shortest vector} in an ideal lattice correspond to the \textit{shortest generator} (SVSG problem for short)? 
    In general, the answer is negative as shown by the following example: In $(4+\sqrt{14})\Z[\sqrt{14}]$, $2$ is the shortest vector, but $2$ is not a generator.
    However, for ideals in $\Z[i]$, $\Z[\zeta_8]$, $\Z[\sqrt{2}]$ or $\Z[\zeta_{16}+\zeta_{16}^7]$, we show that the answer to the corresponding SVSG problem is positive, and obtain the following Theorem~\ref{theorem:shortgenerator}.

    \begin{theorem} $\label{theorem:shortgenerator}$
    For ideals in $\Z[i]$, $\Z[\zeta_8]$, $\Z[\sqrt{2}]$ or $\Z[\zeta_{16}+\zeta_{16}^7]$, finding the shortest vectors in these ideals is equivalent to finding their shortest generators.
    \end{theorem}

    \item    
     Based on the technique of finding the shortest generator, we provide a new method (different from analyzing lattice bases) to analyze the length of the shortest vector for prime ideals lying over $p$ in $\Z[\zeta_{2^{n+1}}]$ with $p\equiv3,5\pmod{8}$. 
    

    \item We explicitly determine  the length of the shortest vector for prime ideals  lying over $p$ with $p\equiv7,9\pmod{16}$ as described in Theorem~\ref{main-theorem}.


    \begin{theorem}\label{main-theorem}
    Suppose~$p\equiv7,9\pmod{16}$~and $\mathfrak{p}$ is a prime ideal of $\Z[\zeta_{2^{n+1}}]$ lying over $p$. Let $(a_{p}, b_{p})$ denote the solution of $a^2-2b^2=p$ in positive integers with minimal $a_{p}$. The length of the shortest vector in $\mathfrak{p}$ is $\sqrt{2^{n}a_{p}}$ under canonical embedding. Further, the length admits the upper bound $\sqrt[4]{2^{2n+1}p}$, which is tighter than the upper bound $2^n\sqrt[4]{p}$ obtained from Minkowski's Theorem.
    \end{theorem}
   
   \end{enumerate}

\section{Preliminaries}
\subsection{Pell equations}
    
    For a positive integer $d$ that is not a square, a generalized Pell equation is an equation in the integer variables $a$ and $b$ of the form
    \[
    a^2 - d b^2 = n,
    \]
    where $n\in\Z-\{0\}$.  
    
    In this paper, the equations $a^2-2b^2=p$ and $a^2-2b^2=-p$ with a prime $p$ are considered. The notation $(a_p,b_p)$ refers to the unique positive integer solution of $a^2 - 2b^2 = p$
    for which $a_p$ is minimal among all positive solutions. Similarly, $(a_{-p}, b_{-p})$ denotes the solution of $a^2-2b^2=-p$ in positive integers with minimal $a_{-p}$. According to \cite[Theorem 3.3]{Conrad2015PELLSEI}, 
    $$a_{p}<\sqrt{2p}.$$

   \subsection{Number fields}
    An algebraic number field $\K$ is a finite extension of the rational number field $\Q$, and such an extension can be constructed by adjoining an algebraic number $\alpha\in\C$ as $\K=\Q(\alpha)$.
    The ring of integers of $\K$ is denoted by $\O_\K$.
This ring of integers  is Dedekind and ideals in $\O_\K$ can be factorized uniquely. Let $\fa$ be an ideal of $\O_\K$, then
we have
\[
\fa=\fp_1^{e_1}\cdots \fp_g^{e_g},
\]
where $\fp_i$ is a prime ideal, $e_i\in \Z$ are non-negative and the factorization is unique up to ordering. The unit group of $\O_\K$, denoted by $\O_\K^\times$, consists of all invertible elements in $\O_\K$. 

\noindent\textbf{Norm of elements and ideals.}
Let $\K$ be a number field of degree $[\K:\Q]=t$, and let
$\varphi_1,\dots,\varphi_t$ be all embeddings of $\K$ into $\C$.
For an element $\alpha\in \K$, the norm of $\alpha$ over $\Q$
is defined as
\[
N_{\K/\Q}(\alpha)
=\prod_{i=1}^t \varphi_i(\alpha).
\]
When the underlying number field is clear from the context,
we will simply write $N(\alpha)$ instead of $N_{\K/\Q}(\alpha)$. Let $\fa$ be a non-zero ideal of $\O_\K$. The norm of $\fa$ is defined as
\[
N(\fa)=|\O_\K/\fa|.
\]
If $\fa$ is a principal ideal generated by $\alpha\in\O_\K$, then $N(\fa)=|N(\alpha)|$.

\noindent\textbf{Cyclotomic fields.} 
A cyclotomic field is a field obtained by adjoining a primitive root of unity to the rational number field.  Let \( m \in \mathbb{Z}^+ \), we define the cyclotomic field as \( \mathbb{Q}(\zeta_m) \), where \( \zeta_m \) is a primitive \( m \)-th root of unity, which can be expressed as:
\[
\zeta_m = e^{2\pi i / m}.
\]
The ring of integers of \( \mathbb{Q}(\zeta_m) \) is known to be \( \mathbb{Z}[\zeta_m] \). For the cyclotomic field \( \mathbb{Q}(\zeta_m) \), an automorphism $ \sigma_i$ with $ (i,m)=1$ is defined by:
\[
\begin{aligned}
\sigma_i:\mathbb{Q}(\zeta_m) &\longrightarrow \mathbb{Q}(\zeta_m) \\
      \zeta_m &\longmapsto \zeta_m^i. \\
\end{aligned}
\]
In this work, we focus on the power-of-two cyclotomic fields where $m=2^n$.

   \subsection{Lattices}
   Given a list of
linearly independent column vectors $\B=(\b_1,\ldots,\b_n)\in \R^{n\times n}$,
the (full rank) \textit{lattice} $\L(\B)$ is the set
\[
\L(\B)=\left\{\sum_{i=1}^{n}x_i\b_i\,|\,x_i\in \Z\right\}.
\]
The \textit{minimum distance} of the lattice is
\[
\lambda_1(\L):=\min_{0\neq v\in\L}||v||
\]
where $||\cdot||$ denotes the Euclidean norm.

Let $\L\in\R^n$ be a full rank lattice. There are several variants of \textit{Shortest Vector Problem}
(SVP) including:
\begin{enumerate}
    \item Search SVP: Given a basis of $\L$, find a non-zero vector $v\in\L$ such that
        \[
        \|v\|=\lambda_1(\L).
        \]
    \item Optimization SVP: Given a basis of $\L$, find $\lambda_1(\L)$.
\end{enumerate}
    According to \cite{regev2004lecture}, the two variants are essentially equivalent.
    
    \subsection{Ideal lattices}
    Let $I$ be a non-zero ideal of $\O_\K$. Then we can obtain ideal lattices by
    the following canonical embedding or coefficient embedding.
    
    \noindent\textbf{Embeddings.} Let $\K$ be a number field of degree $t=[\K:\Q]$ with $r_1$ real embeddings and $r_2$ conjugate pairs of complex embeddings ($t=r_1+2r_2$). The canonical embedding is the injective ring homomorphism:
    \[
    \begin{aligned}
    \Sigma_{\K}:\K &\longrightarrow\R^{r_1}\times\C^{2r_2} \\
    a &\longmapsto (\varphi_1(a),\dots,\varphi_{r_1}(a),\varphi_{r_1+1}(a),\dots,\varphi_{r_1+2r_2}(a)), \\
    \end{aligned}
    \]
    where $\{\varphi_{1},\dots,\varphi_{r_1}\}$ are the real embeddings of $\K$, and $\{\varphi_{r_1+1},\dots,\varphi_{r_1+r_2}\}$ are representatives of the complex embeddings. 
    Suppose $\alpha$ is one of the generators of $\K$, then coefficient embedding maps $\beta=a_0+a_1\alpha+\dots+a_{t-1}\alpha^{t-1}$ to its coefficient vector:
    \[
    \begin{aligned}
    C:\K &\longrightarrow\Q^t \\
     \beta &\longmapsto  (a_0,a_1,\dots,a_{t-1}).\\
     \end{aligned}
    \]
    Under the canonical embedding, the image
    \[
    \Sigma_{\K}(I)=\{\Sigma_{\K}(x)\mid x\in I\}
    \]
    maps into a subspace in $\C^t$ which is isomorphic to $\R^t$ as an inner product space.
   Under the coefficient embedding with respect to a fixed $\Q$-basis
    $\{1,\alpha,\dots,\alpha^{t-1}\}$ of $\K$, the image
    \[
    C(I)=\{C(x)\mid x\in\\I\}\subset \Z^{t}
    \]
    is a full-rank lattice. 
    
    \noindent\textbf{Relationship of two embeddings in power-of-two cyclotomic fields.} Let $\K=\Q(\zeta_{2^{n+1}})$ and $\alpha\in\K$, then the following proposition holds:
    $$||\Sigma_{\K}(\alpha)||=\sqrt{2^n}||C(\alpha)||.$$
    Therefore, an element is shortest in the ideal lattice of $\K$ under the canonical embedding if and only if it is shortest under the coefficient embedding.

    \noindent\textbf{Note:} In this article, unless otherwise specified, when referring to the length of an element or the shortest element, it is with respect to the canonical embedding.
    \subsection{Ideal SVP in cyclotomic fields}
  Pan et al. provided an algorithm to solve general ideal SVP in cyclotomic fields \cite[Algorithm 2]{PXWC21}. Here we restate the result of Pan et al.'s algorithm in the form of a theorem.
    
    Let $\O_{\mathbb{L}}=\Z[\zeta_{2^{n+1}}]$, $\O_{\mathbb{K}}=\Z[\zeta_{2^{k+1}}]$ with $n>k$. Let $I$ be an ideal (not necessarily prime) in $\O_{\K}$. We have 
    $$I\O_{\mathbb{L}}=I\oplus I\zeta_{2^{n+1}}\oplus I\zeta_{2^{n+1}}^2\oplus\cdots\oplus I\zeta_{2^{n+1}}^{n-k-1}.$$
    Therefore, \cite[Algorithm 2]{PXWC21} actually induces the following Theorem~\ref{Theorem:expandPan}.
    \begin{theorem}\label{Theorem:expandPan}
        For an ideal $I$ (not necessarily prime) in $\Z[\zeta_{2^{k+1}}]$, suppose $\alpha$ is the shortest vector in $I$. Then $\alpha$ is also the shortest vector in $I\Z[\zeta_{2^{n+1}}]$ with $n>k$.
    \end{theorem}

    \section{Shortest generator versus shortest vector in related fields} $\label{Section:generator}$

    
    In this section, we prove that for any ideal in $\Z[i]$, $\Z[\zeta_8]$, $\Z[\sqrt{2}]$ and $\Z[\zeta_{16}+\zeta_{16}^7]$, the answer to the corresponding SVSG problem is positive, i.e., there exists a generator that is the shortest vector.
    
    As an application of this property, we present a new method to recompute the lengths of the shortest vectors of prime ideals with $p\equiv3,5\pmod{8}$ as given by Pan et al. In addition, we provide an analysis of the lengths of shortest vectors of prime ideals in $\Z[\sqrt{2}]$ with $p\equiv1,7\pmod{8}$, which leads to an efficient algorithm for computing $a_p$ .

\subsection{ Analysis of ideals in $\Z[i]$ } \label{subsection: Analysis of prime ideal SVP in Z[i]}

     For any ideal in $\Z[i]$, there always exists a generator that is the shortest vector. More specifically, any generator of an ideal in $\Z[i]$ is exactly the shortest vector. 
    \begin{theorem}
        Suppose $I$ is an ideal in $\Z[i]$ and $I=(\alpha)$, then $\alpha$ is the shortest vector in $I$.
    \end{theorem}

    \begin{proof}
        For any non-zero element $\gamma$ in $I$, there exists $0\neq\beta\in\Z[i]$ such that $\gamma=\alpha\beta$. Then 
        $$||\Sigma_{\Q(i)}(\gamma)||^2=2\gamma\overline{\gamma}=2\alpha\overline{\alpha}\beta\overline{\beta}=||\Sigma_{\Q(i)}(\alpha)||^2\cdot|\beta|^2\geq||\Sigma_{\Q(i)}(\alpha)||^2.$$\
        So $\alpha$ is the shortest vector in $I$. 
        
    \end{proof}

     \noindent\textbf{Applications to analyze the length of the shortest vector for prime ideals lying over $p$ in $\Z[\zeta_{2^{n+1}}]$ with $p\equiv5\pmod{8}$.} For any prime ideal $\mathfrak{p}$ in $\Z[\zeta_{2^{n+1}}]$ with $p\equiv5\pmod{8}$,
    Pan et al. demonstrated that $$\lambda_1(C(\mathfrak{p}))=\sqrt{p}$$ by analyzing the lattice basis of $C(\mathfrak{p})$, thus $$\lambda_1(\Sigma_{\Q(i)}(\mathfrak{p}))=\sqrt{2p}.$$ 
    Here we get the same conclusion by analyzing the shortest generator.
     For any $p\equiv5\pmod{8}$, we have 
     \[
     p=a^2+b^2=(a+bi)(a-bi),
     \]
     and $(a+bi)$ is a prime ideal lying over $p$ in $\Z[i]$. We conclude that 
     $$\lambda_1\left(\Sigma_{\Q(i)}((a+bi))\right)=||\Sigma_{\Q(i)}(a+bi))||=\sqrt{2p}.$$
     Since $(a+bi)$ does not split in $\Z[\zeta_{2^{n+1}}]$, according to Theorem~\ref{Theorem:expandPan}, 
     $$\lambda_1(\Sigma_{\Q(\zeta_{2^{n+1}})}((a+bi)\Z[\zeta_{2^{n+1}}]))=||\Sigma_{\Q(\zeta_{2^{n+1}})}(a+bi))||=\sqrt{2^np}.$$
     

\subsection{Analysis of ideals in $\Z[\zeta_8]$}
    In Theorem~\ref{Theorem:Z[8] generator}, we show that for any ideal in $\Z[\zeta_8]$, there always exists a generator that is the shortest vector. 
    Let $\epsilon=1+\sqrt{2}$ and recall that $\Z[\zeta_8]^{\times}=\langle \epsilon,\zeta_8 \rangle$.
    To prove this theorem, we need  Lemma~\ref{lemma:leqEquation}.
    
    \begin{lemma}\label{lemma:leqEquation}
        Let $a,b,t\in \R_{>0}, n\in\Z$. Set $f(n)=at^n+bt^{-n}$, then $\min_{n\in\Z}(f(n))\leq(\sqrt{t}+\frac{1}{\sqrt{t}})(\sqrt{ab})$.
    \end{lemma}
    \begin{proof}
        Let $f(x)=at^x+bt^{-x}$. We have 
        $$f'(x)=\ln{t}(at^x-bt^{-x})=\ln{t}\frac{at^{2x}-b}{t^{2x}}.$$
        Therefore, $f(x)$ takes its minimum value at $x_0=\frac{1}{2}\log_{t}(\frac{b}{a})$. Let $n_0$ be the integer closest to $x_0$, so that $\lvert x_0 - n_0 \rvert \leq \frac{1}{2}$. Then
        \begin{align*}
        \min_{n\in\Z}(f(n))&\leq f(n_0)\\
        &=at^{n-x_0+x_0}+bt^{-(n-x_0+x_0)}\\
        &=\sqrt{ab}(t^{n-x_0}+t^{x_0-n})\\
        &\leq\sqrt{ab}(t^{\frac{1}{2}}+t^{-\frac{1}{2}}).
        \end{align*}
    \end{proof}

    \begin{theorem} \label{Theorem:Z[8] generator}
        Suppose $I$ is an ideal in $\Z[\zeta_8]$, then there exists a generator of $I$ which is the shortest vector in $I$.
    \end{theorem}

    \begin{proof}
        Let $I=(\alpha)$. For any non-zero and non-unit element $\gamma$ in $I$, there exists a non-unit and non-zero $\beta\in\Z[\zeta_8]$ such that $\gamma=\alpha\beta$. We have
        $$||\Sigma_{\Q(\zeta_8)}(\gamma)||^2=2|\gamma|^2+2|\sigma_3({\gamma)}|^2\geq4\sqrt{\sigma_1(\gamma)\sigma_7(\gamma)\sigma_3(\gamma)\sigma_5(\gamma)}\geq4\sqrt{2}\sqrt{N(I)}.$$
        
        Now we prove that there exists a generator $\alpha'$ of $I$ such that $||\Sigma_{\Q(\zeta_8)}(\alpha')||^2\leq4\sqrt{2}\sqrt{N(I)}$.  According to Lemma~\ref{lemma:leqEquation}, there exists $n\in\Z$ such that
        \begin{align*}
            ||\Sigma_{\Q(\zeta_8)}(\alpha\epsilon^n)||^2&=2|\alpha|^2(\epsilon^2)^n+2|\sigma_3({\alpha})|^2(\epsilon^2)^{-n}\\
            &\leq 2\sqrt{N(I)}(\epsilon+\epsilon^{-1})\\
            &=4\sqrt{2}\sqrt{N(I)}. 
        \end{align*}
    \end{proof}

    \noindent\textbf{Applications to analyze the length of the shortest vector for prime ideals lying over $p$ in $\Z[\zeta_{2^{n+1}}]$ with $p\equiv3\pmod{8}$.} For any prime ideal $\mathfrak{p}$ in $\Z[\zeta_{2^{n+1}}]$ with $p\equiv3\pmod{8}$,
    Pan et al. demonstrated that $$\lambda_1(C(\mathfrak{p}))=\sqrt{p}$$ by analyzing the lattice basis of $C(\mathfrak{p})$, thus $$\lambda_1(\Sigma_{\Q(\zeta_8)}(\mathfrak{p}))=\sqrt{4p}.$$ Here we get the same conclusion by analyzing the shortest generator.
    For any $p\equiv3\pmod{8}$, we have
    \[
    p=(a+b\sqrt{-2})(a-b\sqrt{-2}) 
    \]
    and $(a+b\sqrt{-2})$ is a prime ideal lying over $p$ in $\Z[\zeta_8]$. Let $\alpha=a+b\sqrt{-2}$, then all generators in $(\alpha)$ can be written as $\alpha\zeta_8^k\epsilon^n$ where $k,n\in\Z$. Then 
    $$||\Sigma_{\Q(\zeta_8)}(\alpha\zeta_8^k\epsilon^n)||^2=2(a^2+2b^2)(\epsilon^{2n}+\epsilon^{-2n})\geq4(a^2+2b^2)=4p.$$
    Therefore
    $$\lambda_1\left(\Sigma_{\Q(\zeta_8)}((\alpha))\right)=||\Sigma_{\Q(\zeta_8)}(\alpha)||=\sqrt{4p}.$$
    Since $(a+b\sqrt{-2})$ is inert in $\mathbb{Z}[\zeta_{2^{n+1}}]$, according to Theorem~\ref{Theorem:expandPan},
    $$\lambda_1\left(\Sigma_{\Q(\zeta_{2^{n+1}})}((\alpha))\right)=||\Sigma_{\Q(\zeta_{2^{n+1}})}(\alpha)||=\sqrt{2^np}.$$
    
    
\subsection{Analysis of ideal in $\Z[\sqrt{2}]$}
    Like $\Z[i]$ and $\Z[\zeta_8]$, for an ideal in $\Z[\sqrt{2}]$, there exists a generator that is the shortest vector. 
    \begin{theorem}  \label{Theorem:Z[2] generator}
        Suppose $I$ is an ideal in $\Z[\sqrt{2}]$, then there exists a generator of $I$ which is the shortest vector in $I$.
    \end{theorem}

    \begin{proof}
        Let $I=(\alpha)$. For any non-zero and non-unit element $\gamma$ in $I$, there exists a non-unit and non-zero $\beta\in\Z[\sqrt{2}]$ such that $\gamma=\alpha\beta$. We have
        $$||\Sigma_{\Q(\sqrt{2})}(\gamma)||^2=\gamma^2+\tau_2({\gamma})^2\geq2|N(\gamma)|\geq4N(I),$$
        where $\tau_2(\sqrt{2})=-\sqrt{2}$.
        
        Now we prove that there exists a generator $\alpha'$ such that $||\Sigma_{\Q(\sqrt{2})}(\alpha')||^2<4N(I)$. Let $\epsilon=1+\sqrt{2}$. Recall $\Z[\sqrt{2}]^{\times}=\langle\epsilon,-1\rangle$. According to Lemma~\ref{lemma:leqEquation}, there exists $n\in\Z$ such that
        \[||\Sigma_{\Q(\sqrt{2})}(\alpha\epsilon^n)||^2=\alpha^2(\epsilon^2)^n+\tau_2(\alpha)^2(\epsilon^2)^{-n}\leq N(I)(\epsilon+\epsilon^{-1})<4N(I). \]
    \end{proof}

    \noindent\textbf{Applications to analyze the length of the shortest vector for prime ideals lying over $p$ in $\Z[\sqrt{2}]$ with $p\equiv1,7\pmod{8}$.} 
     \begin{lemma}\label{lemma:relation of p and -p}
        Suppose $p$ is a prime, then $a_{p}\geq2b_{p}$ and $a_{-p}=a_{p}-2b_{p}$.
    \end{lemma}

    \begin{proof}
            Note that $a_{p}>b_{p}\sqrt{2}$. Let 
    \[
    a'=|3a_{p}-4b_{p}|=3a_{p}-4b_{p},b'=|3b_{p}-2a_{p}|,
    \]
    such that $a'^2-2b'^2=p$. If $a_{p}<2b_{p}$, we have:
    \begin{itemize}
        \item If $3b_{p}-2a_{p}\geq0$, then $b'=3b_{p}-2a_{p}<(3-2\sqrt{2})b_{p}<b_{p}$.
        \item If $3b_{p}-2a_{p}<0$, then $b'=2a_{p}-3b_{p}<b_{p}$.
    \end{itemize}
    This is contradictory to that $b_{p}$ is the minimal positive root of the equation. So $a_{p}\geq 2b_{p}$.

        $N(a_p-b_p\sqrt{2})=(a+b\sqrt{2})(a-b\sqrt{2})=p$. $\Z[\sqrt{2}]^{\times}=\langle1+\sqrt{2},-1\rangle$. Let $\beta=(1+\sqrt{2})^{2m+1}(-1)^k$ with $m,k\in\Z$, then $N(\beta)=-1$ and $N((a_p-b_p\sqrt{2})\beta)=-p$. Therefore, all positive $a$ such that $a^2-2b^2=-p$ can be represented as
        \begin{align*}
            X_m &=\left|\frac{(a_p-b_p\sqrt{2})(1+\sqrt{2})^{2m+1}+(a_p+b_p\sqrt{2})(1-\sqrt{2})^{2m+1}}{2}\right|\\
                &=\frac{(a_p-b_p\sqrt{2})(1+\sqrt{2})^{2m+1}-(a_p+b_p\sqrt{2})(\sqrt{2}-1)^{2m+1}}{2},
        \end{align*}
        with $m\in\Z$.

        If $m\geq0$, $X_m\geq X_0=a_p-2b_p$.

        If $m<0$, we have
        \begin{align*}
            X_m&=\left|\frac{(a_p-b_p\sqrt{2})(\sqrt{2}-1)^{-2m-1}-(a_p+b_p\sqrt{2})(\sqrt{2}+1)^{-2m-1}}{2}\right|\\
               &=\frac{(a_p+b_p\sqrt{2})(\sqrt{2}+1)^{-2m-1}-(a_p-b_p\sqrt{2})(\sqrt{2}-1)^{-2m-1}}{2}.
        \end{align*}
        Then we have $X_m\geq X_{-1}=a_p+2b_p\geq a_p-2b_p$.
    \end{proof}
    \begin{theorem}
        Suppose $p\equiv 1,7\pmod{8}$ and $\mathfrak{p}$ is a prime ideal of $\Z[\sqrt{2}]$ lying over $p$, then 
        $$\lambda_1\left(\Sigma_{\Q(\sqrt{2})}(\mathfrak{p})\right)=\sqrt{2}\min\left(\sqrt{2a_{p}^2-p},\sqrt{6a_{p}^2-4\sqrt{2}a_{p}\sqrt{a_{p}^2-p}-3p}\right).$$
    \end{theorem}
    \begin{proof}
        
        Let $p=(a+b\sqrt{2})(a-b\sqrt{2})$, and $(a+b\sqrt{2})$ is a prime ideal in $\Z[\sqrt{2}]$. Without loss of generality, we suppose $\mathfrak{p}=(a+b\sqrt{2})$. According to Theorem~\ref{Theorem:Z[2] generator}, the shortest generator in $\mathfrak{p}$ is also the shortest vector in $\mathfrak{p}$.
    
        All generators in $\mathfrak{p}$ can be written as 
        \[
        \alpha_{n,k}=(a+b\sqrt{2})(1+\sqrt{2})^{n}(-1)^k=x_{n,k}+y_{n,k}\sqrt{2}, 
        \]
        where $n\in\Z$, $k\in\{0,1\}$. Then 
        $$N(\alpha_{n,k})=(-1)^np=x_{n,k}^2-2y_{n,k}^2,$$
        which implies that $(x_{n,k},y_{n,k})$ is a root for equation $a^2-2b^2=p$ or $a^2-2b^2=-p$. Compute 
        \begin{align*}
            ||\Sigma_{\Q(\sqrt{2})}(\alpha_{n,k})||^2&=2(x_{n,k}^2+2y_{n,k}^2)\\
            &=\begin{cases}
            2(2x_{n,k}^2-p), &\text{if $n$ is even}, \\
            2(2x_{n,k}^2+p), &\text{if $n$ is odd}.
            \end{cases}
        \end{align*}
        Therefore, smaller $|x_{n,k}|$ implies smaller $||\Sigma_{\Q(\sqrt{2})}(\alpha_{n,k})||$. We conclude that 
        $$\lambda_1\left(\Sigma_{\Q(\sqrt{2})}(\mathfrak{p})\right)=\sqrt{2}\min\left(\sqrt{2a_{p}^2-p},\sqrt{2a_{-p}^2+p}\right).$$
        Based on Lemma~\ref{lemma:relation of p and -p}, we have
        \[
            \sqrt{2a_{-p}^2+p}=\sqrt{6a_{p}^2-4\sqrt{2}a_{p}\sqrt{a_{p}^2-p}-3p}. 
        \]
    \end{proof}

    \noindent\textbf{The upper bound of $\lambda_1\left(\Sigma_{\Q(\sqrt{2})}(\mathfrak{p})\right)$.} Let $$f(a)=\min\left(\sqrt{2a^2-p},\sqrt{6a^2-4\sqrt{2}a\sqrt{a^2-p}-3p}\right),$$ then we have
    $$
    f(a)=\begin{cases}\sqrt{2a^2-p},& \text{if }a<\sqrt{\frac{\sqrt{2}+1}{2}p};\\\sqrt{6a^2-4\sqrt{2}a\sqrt{a^2-p}-3p},&\text{if }a\geq\sqrt{\frac{\sqrt{2}+1}{2}p}.\end{cases}\\
    $$
    Observe 
    \[
    \lambda_1\left(\Sigma_{\Q(\sqrt{2})}(\mathfrak{p})\right)=\sqrt{2}(f(a))\leq \sqrt{2\sqrt{2}p}.
    \]
    
    \noindent\textbf{An algorithm to find $a_{p}$}. Through the above discussion, we can see that for the shortest vector $u+v\sqrt{2}$ in $\mathfrak{p}$, $|u|$ is either $a_p$ or $a_{-p}$. Based on Lemma~\ref{lemma:relation of p and -p}, $a_p$ can be computed by $|u|$ and $|v|$. Therefore, we design Algorithm~\ref{Algorithm-ap} to find $a_{p}$.
   \begin{algorithm}[H]
      \caption{An algorithm to find $a_{p}$}\label{Algorithm-ap} 
      \begin{algorithmic}[1]
        \Require
          $p$: a prime $p$ satisfying $p\equiv\pm1\pmod{8}$.  
        \Ensure
           $a_{p}$: the minimal positive integer $a_{p}$ such that $a^2-2b^2=p$.
           \State Find a root of $x^2\equiv2\pmod{p}$, denoted by $a$.
           \State Construct $\B:=\begin{pmatrix}
            p & p \\
            a+\sqrt{2} & a-\sqrt{2}
            \end{pmatrix}$.
            \State Find the shortest vector in $\L(\B)$: $(u+v\sqrt{2},u-v\sqrt{2})=\text{LLL}(\B)$.
            \If {$u^2-2v^2=p$}  \State \Return $|u|$.
            \ElsIf {$u^2-2v^2=-p$}   \State \Return $2|v|-|u|$.
            \EndIf
      \end{algorithmic}
    \end{algorithm}
    If $p\equiv7\pmod{8}$, $a=2^{\frac{p+1}{4}}\mod{p}$. If $p\equiv1\pmod{8}$, there are several efficient algorithms to find $a$, such as Tonelli-Shanks~\cite{shanks1971class} and Cipolla~\cite{cipolla1903metodo}.
    

\subsection{Analysis of ideal in $\Z[\zeta_{16}+\zeta_{16}^7]$}\label{Sub:generator Z1616}

        \begin{theorem}\label{Theorem:Z[16167] generator}
        Suppose $I$ is an ideal in $\Z[\zeta_{16}+\zeta_{16}^7]$, then there exists a generator of $I$ which is the shortest vector in $I$.
    \end{theorem}

    \begin{proof}
        Let $I=(\alpha)$. For any non-zero and non-unit element $\gamma$ in $I$, there exists a non-unit and non-zero $\beta\in\Z[\zeta_{16}+\zeta_{16}^7]$ such that $\gamma=\alpha\beta$. We have
        $$||\Sigma_{\Q(\zeta_{16}+\zeta_{16}^7)}(\gamma)||^2=2|\gamma|^2+2|\sigma_3({\gamma})|^2\geq4\sqrt{2}\sqrt{N(I)}.$$
        $\Z[\zeta_{16}+\zeta_{16}^7]^{\times}=\langle1+\sqrt{2},-1\rangle$. Then, as in the proof of Theorem~\ref{Theorem:Z[8] generator}, there exists a generator $\alpha'$ such that 
        \[
       ||\Sigma_{\Q(\zeta_{16}+\zeta_{16}^7)}(\alpha')||^2<4\sqrt{2}\sqrt{N(I)}. 
        \]
    \end{proof}

\section{SVP for prime ideals lying over $p\equiv9\pmod{16}$ in $\Z[\zeta_{2^{n+1}}]$}

In this section, we analyze the optimization SVP for prime ideals in $\Z[\zeta_{2^{n+1}}]$ when $p\equiv9\pmod{16}$. We also give an upper bound for the length which is tighter than the upper bound obtained from Minkowski's Theorem.

\subsection{High level idea of the analysis when $p\equiv9\pmod{16}$.} 

If $p\equiv9\pmod{16}$, prime ideals lying over $p$ in $\Z[\sqrt{2}]$ split into two prime ideals in $\Z[\zeta_8]$, and prime ideals in $\Z[\zeta_8]$ remain inert in $\Z[\zeta_{2^{n+1}}]$. The decomposition of prime ideals is illustrated in Figure~\ref{Fig-9-16}. 

\begin{figure}[h]
  \centering
    \begin{tikzpicture}
        \node (1) at (0,0) {$\Q$};
        \node (2) at (0,1) {$\Q(\sqrt{2})$};
        \node (3) at (0,2) {$\Q(\zeta_8)$};
        \node (4) at (0,3) {$\Q(\zeta_{2^{n+1}})$};
        \draw[-] (1) -- (2) -- (3);
        \draw[line width=1pt,
        dash pattern=on 0pt off 4pt,
        dash phase=3pt,
        line cap=round] (3) -- (4);

        \node (5) at (2.5,0) {$(p)$};
        \node (6) at (1.5,1) {$\mathfrak{p}_1$};
        \node (7) at (3.5,1) {$\mathfrak{p}_2$};
        \node (8) at (1,2) {$\mathfrak{p}_{11}$};
        \node (9) at (2,2) {$\mathfrak{p}_{12}$};
        \node (10) at (3,2) {$\mathfrak{p}_{21}$};
        \node (11) at (4,2) {$\mathfrak{p}_{22}$};
        \draw[-] (5)--(6)--(8);
        \draw[-] (5)--(7)--(11);
        \draw[-] (6)--(9);
        \draw[-] (7)--(10);

        \node (12) at (1,3) {$\tilde{\mathfrak{p}}_{11}$};
        \node (13) at (2,3) {$\tilde{\mathfrak{p}}_{12}$};
        \node (14) at (3,3) {$\tilde{\mathfrak{p}}_{21}$};
        \node (15) at (4,3) {$\tilde{\mathfrak{p}}_{22}$};
        \draw[line width=1pt,
        dash pattern=on 0pt off 4pt,
        dash phase=3pt,
        line cap=round] (8) -- (12);

        \draw[line width=1pt,
        dash pattern=on 0pt off 4pt,
        dash phase=3pt,
        line cap=round] (9) -- (13);

        \draw[line width=1pt,
        dash pattern=on 0pt off 4pt,
        dash phase=3pt,
        line cap=round] (10) -- (14);

        \draw[line width=1pt,
        dash pattern=on 0pt off 4pt,
        dash phase=3pt,
        line cap=round] (11) -- (15);
    \end{tikzpicture}
    \caption{Decomposition of prime ideals when $p\equiv9\pmod{16}$ }\label{Fig-9-16}
\end{figure}

\begin{ex}
    For rational prime 89, the decomposition is illustrated in Figure~\ref{Fig-89}.

    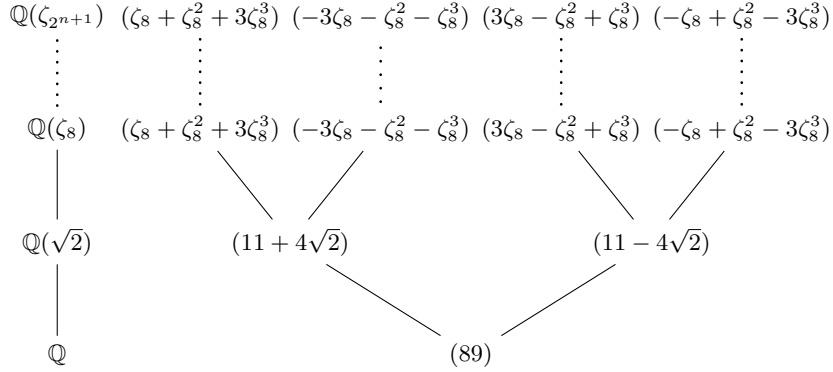
\begin{figure}[H]
  \centering
    \begin{tikzpicture}
        \node (1) at (0.5,0) {$\Q$};
        \node (2) at (0.5,1.5) {$\Q(\sqrt{2})$};
        \node (3) at (0.5,3) {$\Q(\zeta_8)$};
        \node (4) at (0.5,4.5) {$\Q(\zeta_{2^{n+1}})$};
        \draw[-] (1) -- (2) -- (3);
        \draw[line width=1pt,
        dash pattern=on 0pt off 4pt,
        dash phase=3pt,
        line cap=round] (3) -- (4);

        \node (5) at (6,0) {$(89)$};
        \node (6) at (3.6,1.5) {$(11+4\sqrt{2})$};
        \node (7) at (8.4,1.5) {$(11-4\sqrt{2})$};
        \node (8) at (2.4,3) {$(\zeta_8+\zeta_8^2+3\zeta_8^3)$};
        \node (9) at (4.8,3) {$(-3\zeta_8-\zeta_8^2-\zeta_8^3)$};
        \node (10) at (7.2,3) {$(3\zeta_8-\zeta_8^2+\zeta_8^3)$};
        \node (11) at (9.6,3) {$(-\zeta_8+\zeta_8^2-3\zeta_8^3)$};
        \draw[-] (5)--(6)--(8);
        \draw[-] (5)--(7)--(11);
        \draw[-] (6)--(9);
        \draw[-] (7)--(10);

        \node (12) at (2.4,4.5) {$(\zeta_8+\zeta_8^2+3\zeta_8^3)$};
        \node (13) at (4.8,4.5) {$(-3\zeta_8-\zeta_8^2-\zeta_8^3)$};
        \node (14) at (7.2,4.5) {$(3\zeta_8-\zeta_8^2+\zeta_8^3)$};
        \node (15) at (9.6,4.5) {$(-\zeta_8+\zeta_8^2-3\zeta_8^3)$};
        \draw[line width=1pt,
        dash pattern=on 0pt off 4pt,
        dash phase=3pt,
        line cap=round] (8) -- (12);

        \draw[line width=1pt,
        dash pattern=on 0pt off 5pt,
        dash phase=3pt,
        line cap=round] (9) -- (13);

        \draw[line width=1pt,
        dash pattern=on 0pt off 4pt,
        dash phase=3pt,
        line cap=round] (10) -- (14);

        \draw[line width=1pt,
        dash pattern=on 0pt off 4pt,
        dash phase=3pt,
        line cap=round] (11) -- (15);
    \end{tikzpicture}
    \caption{Decomposition of prime ideals for $p=89$}\label{Fig-89}
\end{figure}
    
\end{ex}
According to Theorem~\ref{Theorem:expandPan}, as long as we find the shortest vector $\alpha$ of a prime ideal $\mathfrak{p}$ lying over $p$ in $\Z[\zeta_8]$, $\alpha$ is also the shortest vector in the corresponding prime ideal in $\Z[\zeta_{2^{n+1}}]$. As discussed in Section~\ref{Section:generator}, to find $\alpha$, we only need to find the shortest generator in $\mathfrak{p}$. The generators of prime ideals in $\mathbb{Z}[\zeta_8]$ admit a factorization relation with the generators of the corresponding prime ideals in $\Z[\sqrt{2}]$. Thus the length of the generators in $\mathfrak{p}\cap\Z[\sqrt{2}]$ helps us understand the length of the generators in $\mathfrak{p}$.
\subsection{Detailed analysis}
The main result of the analysis is described in Theorem~\ref{theorem: p916}. Before that, we need the following Lemma~\ref{lem:u2k} and Lemma~\ref{p-1-8}.
\begin{lemma}\label{lem:u2k}
    If $u\in\Z[\sqrt{2}]^{\times}$ is totally positive, then $u=(1+\sqrt{2})^{2k},k\in\Z$.
\end{lemma}
\begin{proof}
    $\Z[\sqrt{2}]^{\times}=\langle-1,1+\sqrt{2}\rangle$. Then $u=(1+\sqrt{2})^{2k},k\in\Z$, otherwise it contradicts that $u$ is totally positive.
\end{proof}

\begin{lemma}\label{p-1-8}
    For a rational prime $p\equiv1\pmod{8}$, there exists $a,b\in\Z^+$ such that
    \[
     p=(a+b\sqrt{2})(a-b\sqrt{2}).
     \]
     Then there exists $\alpha\in\Z[\zeta_8]$ such that 
     \[
     a+b\sqrt{2}=\sigma_1(\alpha)\sigma_7(\alpha),a-b\sqrt{2}=\sigma_3(\alpha)\sigma_5(\alpha),
     \]
     and 
     \[
     ||\Sigma_{\Q(\zeta_8)}(\alpha)||^2=2\sigma_1(\alpha)\sigma_7(\alpha)+2\sigma_3(\alpha)\sigma_5(\alpha)=4a. 
     \]
\end{lemma}

\begin{proof}
    In $\Z[\sqrt{2}]$, we have
    \[
    (p)=(a+b\sqrt{2})(a-b\sqrt{2}). 
    \]
    In $\Z[\zeta_8]$, $(a+b\sqrt{2})$ splits two conjugate principal prime ideals:
    $$(a+b\sqrt{2})=(\sigma_1(\alpha'))(\sigma_7(\alpha')),\alpha'\in\Z[\zeta_{8}].$$
    Therefore, there exists a unit $u\in\Z[\zeta_8]^{\times}$ such that 
    \[
    a+b\sqrt{2}=u\sigma_1(\alpha')\sigma_7(\alpha').
    \]
    Note that 
    \[
    p=(a+b\sqrt{2})(a-b\sqrt{2})
    \]
    implies that $a+b\sqrt{2}$ is totally positive. Also, note that 
    \[
    \sigma_1(\alpha')\sigma_7(\alpha')=N_{\Q(\zeta_8)/\Q(\sqrt{2})}(\sigma_1(\alpha'))\in\Z[\sqrt{2}] 
    \]
    and we have
    \begin{align*}
    &\sigma_1(\alpha')\sigma_7(\alpha')=|\sigma_1(\alpha')|^2>0,\\
    &\sigma_3(\sigma_1(\alpha')\sigma_7(\alpha'))=\sigma_3(\alpha')\sigma_5(\alpha')=|\sigma_3(\alpha')|^2>0.
    \end{align*}
    So $\sigma_1(\alpha')\sigma_7(\alpha')$ is totally positive. Then $u\in\Z[\sqrt{2}]^\times$ and $u$ is totally positive. According to Lemma \ref{lem:u2k}, $u=(1+\sqrt{2})^{2k},k\in\Z$. Let $\alpha=(1+\sqrt{2})^k\alpha'$, then 
    \[
    a+b\sqrt{2}=\sigma_1(\alpha)\sigma_7(\alpha), a-b\sqrt{2}=\sigma_3(a+b\sqrt{2})=\sigma_3(\alpha)\sigma_5(\alpha). 
    \]
\end{proof}


\begin{theorem}\label{theorem: p916}
    Suppose $p\equiv9\pmod{16}$ and $\mathfrak{p}$ is a prime ideal of $\Z[\zeta_{2^{n+1}}]$ lying over $p$, then
    $$\lambda_1\left(\Sigma_{\Q(\zeta_{2^{n+1}})}(\mathfrak{p})\right)=\sqrt{2^na_p}.$$
\end{theorem}

\begin{proof}
    According to Lemma~\ref{p-1-8}, there exists
    $\alpha\in\Z[\zeta_8]$ such that 
    $$ p=(a_{p}+b_{p}\sqrt{2})((a_{p}-b_{p}\sqrt{2}))=\sigma_1(\alpha)\sigma_7(\alpha)\sigma_3(\alpha)\sigma_5(\alpha).$$
    In $\Z[\zeta_8]$, $(p)$ splits into 4 conjugate prime ideals:
    $$(p)=(\sigma_1(\alpha))(\sigma_7(\alpha))(\sigma_3(\alpha))(\sigma_5(\alpha)).$$
    Without loss of generality, we suppose $\mathfrak{p}\cap\Z[\zeta_8]=(\sigma_1(\alpha))$. According to Theorem~\ref{Theorem:expandPan}, the shortest element in $\mathfrak{p}$ is the same as the shortest element in $(\sigma_1(\alpha))$. Since there exists a generator of $(\sigma_1(\alpha))$ that is the shortest, in the following we prove
    \[
    \lambda_1\left(\Sigma_{\Q(\zeta_8)}((\sigma_1(\alpha)))\right)=\sqrt{4a_{p}}
    \]
    by determining that the length of the shortest generator is $\sqrt{4a_p}$.

    Note that $\Z[\zeta_8]^{\times}=\langle1+\sqrt{2},\zeta_8\rangle$. All generators in $(\sigma_1(\alpha))$ can be written as $\alpha\zeta_8^k(1+\sqrt{2})^n$ where $k,n\in\Z$. Then 
    \begin{align*}
        ||\Sigma_{\Q(\zeta_8)}(\alpha\zeta_8^k(1+\sqrt{2})^n)||^2&=2|\alpha\zeta_8^k(1+\sqrt{2})^n|^2+2|\sigma_3(\alpha\zeta_8^k(1+\sqrt{2})^n)|^2\\
        &=2(1+\sqrt{2})^{2n}|\alpha|^2+2(1-\sqrt{2})^{2n}|\sigma_3(\alpha)|^2\\
        &=2(1+\sqrt{2})^{2n}(a_p+b_p\sqrt{2})+2(1-\sqrt{2})^{2n}(a_p-b_p\sqrt{2}).
    \end{align*}
    
    Let $x_n+y_n\sqrt{2}=(1+\sqrt{2})^{2n}(a_p+b_p\sqrt{2})$. Note that $x_n>0,y_n>0$ and $x_n+y_n\sqrt{2}$ is totally positive. Then 
    $$||\Sigma_{\Q(\zeta_8)}(\alpha\zeta_8^k(1+\sqrt{2})^n)||^2=4x_n.$$
    Since
    \[
    N_{\Q(\sqrt{2})}(x_n+y_n\sqrt{2})=(1+\sqrt{2})^{2n}(1-\sqrt{2})^{2n}(a_p^2-2b_p^2)=p, 
    \]
    this implies that $(x_{n},y_{n})$ is a root of $a^2-2b^2=p$. So, $x_n\geq a_p$.

    Finally, we conclude that $\alpha$ is the shortest vector in $(\sigma_1(\alpha))$ with length $\sqrt{4a_p}$, and then $\alpha$ is the shortest vector in $\mathfrak{p}$ with length $\sqrt{2^{n}a_p}$. 
    
    \end{proof}

  Based on Algorithm~\ref{Algorithm-ap} and Theorem~\ref{theorem: p916}, we can obtain more efficient classical and quantum algorithms for computing $\lambda_1\left(\Sigma_{\Q(\zeta_{2^{n+1}})}(\mathfrak{p})\right)$.
    \subsection{The upper bound of the length of the shortest vector} \label{subsection:upperBoundp916}
    Recall that $a_{p}<\sqrt{2p}$, then according to Theorem~\ref{theorem: p916}, we can deduce an upper bound of $\lambda_1\left(\Sigma_{\Q(\zeta_{2^{n+1}})}(\mathfrak{p})\right)$.
    \begin{corollary}
     Suppose $p\equiv9\pmod{16}$ and $\mathfrak{p}$ is a prime ideal of $\Z[\zeta_{2^{n+1}}]$ lying over $p$, then 
     $$\lambda_1\left(\Sigma_{\Q(\zeta_{2^{n+1}})}(\mathfrak{p})\right)<\sqrt[4]{2^{2n+1}p}.$$
    \end{corollary}
    As comparison, Minkowski's Theorem implies an upper bound as follows:
    $$\lambda_1\left(\Sigma_{\Q(\zeta_{2^{n+1}})}(\mathfrak{p})\right)<\sqrt{2^{n}}\det(\mathfrak{p})^{\frac{1}{2^n}}=\sqrt{2^{n}}\left(N(\mathfrak{p})\sqrt{\Delta\Q({\zeta_{2^{n+1}}})}\right)^{\frac{1}{2^n}}
    =2^n\sqrt[4]{p}.$$
    Note that $\sqrt[4]{2^{2n+1}p}=2^{\frac{2n+1}{4}}\sqrt[4]{p}<2^n\sqrt[4]{p}$. So the upper bound we propose is tighter than the upper bound obtained from Minkowski's Theorem.
    \begin{ex}
    $\mathfrak{p}=(\zeta_8+\zeta_8^2+3\zeta^3)\Z[\zeta_8]$ is a prime ideal lying over $89$. $\zeta_8+\zeta_8^2+3\zeta^3$ is the shortest vector in $\mathfrak{p}$ with length $2\sqrt{11}\approx6.33$. The Minkowski bound for $\mathfrak{p}$ is $4\sqrt[4]{89}\approx12.29$. The new upper bound we show is $\sqrt[4]{2^589}\approx7.31$.
    \end{ex}
    
\section{SVP for prime ideals lying over $p\equiv7\pmod{16}$ in $\Z[\zeta_{2^{n+1}}]$}
In this section, we analyze the length of the shortest vector for prime ideals in $\Z[\zeta_{2^{n+1}}]$ when $p\equiv7\pmod{16}$, and give an upper bound for it.
\subsection{High level idea of the analysis when $p\equiv7\pmod{16}$.} 

If $p\equiv7\pmod{16}$, prime ideals lying over $p$ in $\Z[\sqrt{2}]$ split into two prime ideals in $\Z[\zeta_{16}+\zeta_{16}^7]$. They remain inert and do not split in $\Z[\zeta_{2^{n+1}}]$, as illustrated in Figure~\ref{Figure:p=7decomposition}.

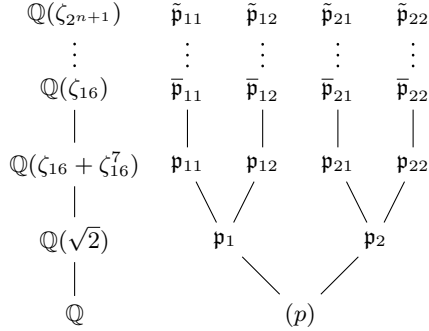
\begin{figure}[h] 
  \centering
    \begin{tikzpicture}
        \node (1) at (-0.5,0) {$\Q$};
        \node (2) at (-0.5,1) {$\Q(\sqrt{2})$};
        \node (3) at (-0.5,2) {$\Q(\zeta_{16}+\zeta_{16}^7)$};
        \node (4) at (-0.5,3) {$\Q(\zeta_{16})$};
        \draw[-] (1) -- (2) -- (3) -- (4);
        \node (16) at (-0.5,4) {$\Q(\zeta_{2^{n+1}})$};
        \draw[line width=1pt,
        dash pattern=on 0pt off 4pt,
        dash phase=3pt,
        line cap=round] (4) -- (16);

        \node (5) at (2.5,0) {$(p)$};
        \node (6) at (1.5,1) {$\mathfrak{p}_1$};
        \node (7) at (3.5,1) {$\mathfrak{p}_2$};
        \node (8) at (1,2) {$\mathfrak{p}_{11}$};
        \node (9) at (2,2) {$\mathfrak{p}_{12}$};
        \node (10) at (3,2) {$\mathfrak{p}_{21}$};
        \node (11) at (4,2) {$\mathfrak{p}_{22}$};
        \draw[-] (5)--(6)--(8);
        \draw[-] (5)--(7)--(11);
        \draw[-] (6)--(9);
        \draw[-] (7)--(10);

        \node (12) at (1,3) {$\overline{\mathfrak{p}}_{11}$};
        \node (13) at (2,3) {$\overline{\mathfrak{p}}_{12}$};
        \node (14) at (3,3) {$\overline{\mathfrak{p}}_{21}$};
        \node (15) at (4,3) {$\overline{\mathfrak{p}}_{22}$};
        \node (17) at (1,4) {$\tilde{\mathfrak{p}}_{11}$};
        \node (18) at (2,4) {$\tilde{\mathfrak{p}}_{12}$};
        \node (19) at (3,4) {$\tilde{\mathfrak{p}}_{21}$};
        \node (20) at (4,4) {$\tilde{\mathfrak{p}}_{22}$};
        \draw[-] (8)--(12);
        \draw[-] (9)--(13);
        \draw[-] (10)--(14);
        \draw[-] (11)--(15);
        
        \draw[line width=1pt,
        dash pattern=on 0pt off 4pt,
        dash phase=3pt,
        line cap=round] (12) -- (17);

        \draw[line width=1pt,
        dash pattern=on 0pt off 4pt,
        dash phase=3pt,
        line cap=round] (13) -- (18);

        \draw[line width=1pt,
        dash pattern=on 0pt off 4pt,
        dash phase=3pt,
        line cap=round] (14) -- (19);

        \draw[line width=1pt,
        dash pattern=on 0pt off 4pt,
        dash phase=3pt,
        line cap=round] (15) -- (20);

    \end{tikzpicture}
    \caption{Decomposition of prime ideals when $p\equiv 7\pmod {16}$}\label{Figure:p=7decomposition}
\end{figure}

Similarly, according to Theorem~\ref{Theorem:expandPan}, for a prime ideal $\mathfrak{p}$ lying over $p$ in $\Z[\zeta_{2^{n+1}}]$, the shortest vector in $\mathfrak{p}$ can be determined through finding the shortest vector in $\mathfrak{p}$ in $\mathfrak{p}\cap\Z[\zeta_{16}]$. But, here we first prove that the shortest vector $\alpha$ in $\mathfrak{p} \cap \Z[\zeta_{16}+\zeta_{16}^7]$ is also the shortest vector in $\mathfrak{p} \cap \Z[\zeta_{16}]$, thereby reducing the problem of finding the shortest vector in $\mathfrak{p}$ to that of finding the shortest vector in $\mathfrak{p} \cap \Z[\zeta_{16}+\zeta_{16}^7]$. In this way, we can exploit the result that in $\mathbb{Z}[\zeta_{16}+\zeta_{16}^7]$ there exists a generator which is the shortest vector (Theorem~\ref{Theorem:Z[16167] generator}), thus simplifying the task of finding the shortest vector.

Like the factorization relation between generators of prime ideals in $\Z[\zeta_8]$ and $\Z[\sqrt{2}]$ , the generators of such prime ideals in $\Z[\zeta_{16}+\zeta_{16}^7]$ also have a factorization relation with the generators of the corresponding prime ideals in $\Z[\sqrt{2}]$. This helps us to understand the length of the shortest vector in corresponding prime ideals.

\subsection{Detailed analysis}

To begin with, we demonstrate that for such prime ideals, the shortest element is preserved when extending from
\(\mathbb{Z}[\zeta_{16}+\zeta_{16}^7]\) to \(\mathbb{Z}[\zeta_{16}]\) in Lemma~\ref{lemma:1616216}. This requires the following Lemma~\ref{Lemma-7-16}.

\begin{lemma}\label{Lemma-7-16}
    For a rational prime $p\equiv7\pmod{16}$, $p=(a+b\sqrt{2})(a-b\sqrt{2})$. Then there exists $\alpha\in\Z[\zeta_{16}+\zeta_{16}^7]$ such that
    \[
    a+b\sqrt{2}=\sigma_1(\alpha)\sigma_{15}(\alpha)=\sigma_7(\alpha)\sigma_9(\alpha),
   \]
   \[
    a-b\sqrt{2}=\sigma_3(\alpha)\sigma_{13}(\alpha)=\sigma_5(\alpha)\sigma_{11}(\alpha).
    \]
\end{lemma}

\begin{proof}
    In $\Z[\sqrt{2}]$, we have 
    \[
    (p)=(a+b\sqrt{2})(a-b\sqrt{2}).
    \]
    In $\Z[\zeta_{16}+\zeta_{16}^7]$, $(a+b\sqrt{2})$ splits into two conjugate principal prime ideals:
    $$(a+b\sqrt{2})=(\sigma_1(\alpha'))(\sigma_{15}(\alpha')),\alpha'\in\Z[\zeta_{16}+\zeta_{16}^7].$$
    Therefore, there exists a unit $u\in\Z[\zeta_{16}+\zeta_{16}^7]^{\times}$ such that
    \[
     a+b\sqrt{2}=u\sigma_1(\alpha')\sigma_{15}(\alpha').
    \]
    Note that $p=(a+b\sqrt{2})(a-b\sqrt{2})$ implies that $a+b\sqrt{2}$ is totally positive. Also, note that
    \[
    \sigma_1(\alpha')\sigma_{15}(\alpha')=N_{\Q(\zeta_{16}+\zeta_{16}^7)/\Q(\sqrt{2})}(\sigma_1(\alpha'))\in\Z[\sqrt{2}]
    \]
     and
    \begin{align*}
    &\sigma_1(\alpha')\sigma_{15}(\alpha')=|\sigma_1(\alpha')|^2>0,\\
    &\sigma_3(\sigma_1(\alpha')\sigma_{15}(\alpha'))=\sigma_3(\alpha')\sigma_{13}(\alpha')=|\sigma_3(\alpha')|^2>0.
    \end{align*}
     Therefore, $\sigma_1(\alpha')\sigma_{15}(\alpha')$ is totally positive. Then $u\in\Z[\sqrt{2}]^\times$ and $u$ is totally positive. According to Lemma \ref{lem:u2k}, $u=(1+\sqrt{2})^{2k},k\in\Z$. Let $\alpha=(1+\sqrt{2})^k\alpha'$, then
    $a+b\sqrt{2}=\sigma_1(\alpha)\sigma_{15}(\alpha)=\sigma_7(\alpha)\sigma_9({\alpha})$. $a-b\sqrt{2}=\sigma_3(a+b\sqrt{2})=\sigma_3(\alpha)\sigma_{13}(\alpha)=\sigma_5(\alpha)\sigma_{11}(\alpha)$.
\end{proof}

\begin{lemma} \label{lemma:1616216}
    Suppose~$p\equiv7\pmod{16}$~and $\mathfrak{p}$ is a prime ideal of $\Z[\zeta_{16}+\zeta_{16}^7]$ lying over $p$. Let $\alpha$ be the shortest vector in $\mathfrak{p}$, then $\alpha$ is also the shortest vector in $\mathfrak{p}\Z[\zeta_{16}]$.
\end{lemma}

\begin{proof}
    According to subsection~\ref{Sub:generator Z1616}, we may assume that $\alpha$ is a unit. Let
    \[
    p=(a_p+b_p\sqrt{2})(a_p-b_p\sqrt{2}).
    \]
     Set $\alpha'\in\Z[\zeta_{16}+\zeta_{16}^7]$ such that
    \[
    \sigma_1(\alpha')\sigma_{15}(\alpha')=a_p+b_p\sqrt{2},
    \]
    then we suppose 
    \[
    \mathfrak{p}=\alpha'\Z[\zeta_{16}+\zeta_{16}^7]
    \]
    and
    \[
    \mathfrak{p}\Z[\zeta_{16}]=\alpha'\Z[\zeta_{16}].
    \]
    We now prove that for any non-zero element $\gamma=\alpha'\beta$ in $\mathfrak{p}\Z[\zeta_{16}]$ with 
    $$\beta=a+b\zeta+c\zeta^2+d\zeta^3+e\zeta^4+f\zeta^5+g\zeta^6+h\zeta^7\in\Z[\zeta_{16}],$$
    it holds that
    $$||\Sigma_{\Q(\zeta_{16})}(\gamma)||^2\geq||\Sigma_{\Q(\zeta_{16})}(\alpha')||^2\geq||\Sigma_{\Q(\zeta_{16})}(\alpha)||^2.$$
    We compute
    \begin{align*}
         & \frac{1}{2}||\Sigma_{\Q(\zeta_{16})}(\gamma)||^2 \\
        =&\sigma_1(\gamma)\sigma_{15}(\gamma)+\sigma_7(\gamma)\sigma_9(\gamma)+\sigma_3(\gamma)\sigma_{13}(\gamma)+\sigma_5(\gamma)\sigma_{11}(\gamma)\\
        =&|\sigma_{1}(\alpha')|^2|\sigma_{1}(\beta)|^2+|\sigma_{7}(\alpha')|^2|\sigma_{7}(\beta)|^2+|\sigma_{3}(\alpha')|^2|\sigma_{3}(\beta)|^2+|\sigma_{5}(\alpha')|^2|\sigma_{5}(\beta)|^2\\
        =&(a_{p}+b_{p}\sqrt{2})(|\sigma_{1}(\beta)|^2+|\sigma_{7}(\beta)|^2)+(a_{p}-b_{p}\sqrt{2})(|\sigma_{3}(\beta)|^2+|\sigma_{5}(\beta)|^2)\\
        =&2(a_{p}+b_{p}\sqrt{2})(x+y\sqrt{2})+2(a_{p}-b_{p}\sqrt{2})(x-y\sqrt{2})\\
        =&4a_{p}x+8b_{p}y,
    \end{align*}
    where 
    \[
    x=a^2 + b^2 + c^2 + d^2 + e^2 + f^2 + g^2 + h^2,
    \]
    \[
    y=ac + bd + ce + df - ag + eg - bh + fh.
    \]
    Since $\beta$ is non-zero, 
    \[
    x+y\sqrt{2}>0,x-y\sqrt{2}>0,x\geq1.
    \]  
    Now we demonstrate that
    \[
    \frac{1}{2}||\Sigma_{\Q(\zeta_{16})}(\gamma)||^2\geq\frac{1}{2}||\Sigma_{\Q(\zeta_{16})}(\alpha')||^2=4a_{p}
    \]
    by considering several cases.
    
    \noindent\textbf{Cases 1:} If $y\geq0$, then we have
    \[
    \frac{1}{2}||\Sigma_{\K}(\gamma)||^2\geq4a_{p}x\geq4a_{p}.
    \]
    
    \noindent\textbf{Cases 2:} If $y<0$ and $x+y\sqrt{2}\geq1$, then we have
    \[
    \frac{1}{2}||\Sigma_{\K}(\gamma)||^2\geq4a_{p}.
    \]

    \noindent\textbf{Cases 3:} If $y<0$, $0< x+y\sqrt{2}<1$ and $1\leq x<4$, then 
    \[
    x+y\sqrt{2}\in \{2-\sqrt{2}, 3-2\sqrt{2} \}.
    \]
    If $x+y\sqrt{2}=2-\sqrt{2}$, recall that $a_p\geq2b_p$, then we have 
    $$\frac{1}{2}||\Sigma_{\K}(\gamma)||^2=8a_{p}-8b_{p}\geq4a_{p}.$$
    If $x+y\sqrt{2}=3-2\sqrt{2}$, we have
    $$\frac{1}{2}||\Sigma_{\K}(\gamma)||^2=12a_{p}-16b_{p}\geq4a_{p}.$$

     \noindent\textbf{Cases 4:} If $y<0$, $0\leq x+y\sqrt{2}<1$ and $x\geq4$, then according to $x+y\sqrt{2}>0$, we have
     $$\frac{1}{2}||\Sigma_{\K}(\gamma)||^2=4a_{p}x+8b_{p}y>4x(a_{p}-b_{p}\sqrt{2})\geq4x(1-\frac{\sqrt{2}}{2})a_{p}>4a_{p}.$$
     Therefore, we have \[||\Sigma_{\Q(\zeta_{16})}(\gamma)||^2\geq||\Sigma_{\Q(\zeta_{16})}(\alpha')||^2\geq||\Sigma_{\Q(\zeta_{16})}(\alpha)||^2. \] 
\end{proof}

Lemma~\ref{lemma:1616216} implies that SVP for a prime ideal $\mathfrak{a}$ of $\mathbb{Z}[\zeta_{16}]$ lying over a prime $p \equiv 7 \pmod{16}$ can be reduced to the SVP in a $4$-dimensional lattice, thereby improving the efficiency of solving the SVP for $\mathfrak{a}$.

In the following Theorem~\ref{theorem:p716}, we discuss the length of the shortest vector in prime ideal of $\Z[\zeta_{2^{n+1}}]$ lying over $p$ when $p\equiv7\pmod{16}$.

\begin{theorem}\label{theorem:p716}
    Suppose~$p\equiv7\pmod{16}$~and $\mathfrak{p}$ is a prime ideal of $\Z[\zeta_{2^{n+1}}]$ lying over $p$, then 
    $$\lambda_1\left(\Sigma_{\Q(\zeta_{2^{n+1}})}(\mathfrak{p})\right)=\sqrt{2^na_p}.$$
\end{theorem}

\begin{proof}

    The problem is reduced to find the shortest generator in $\mathfrak{p}\cap\Z[\zeta_{16}+\zeta_{16}^7]$. According to Lemma~\ref{Lemma-7-16}, there exists $\alpha\in\Z[\zeta_{16}+\zeta_{16}^7]$, such that
    \[p=(a_{p}+b_{p}\sqrt{2})(a_{p}-b_{p}\sqrt{2})=\sigma_1(\alpha)\sigma_{15}(\alpha)\sigma_3(\alpha)\sigma_{13}(\alpha).\] 
    In $\Z[\zeta_{16}+\zeta_{16}^7]$, we have
    \[(p)=(a_{p}+b_{p}\sqrt{2})(a_{p}-b_{p}\sqrt{2})=(\sigma_1(\alpha))(\sigma_{15}(\alpha))(\sigma_3(\alpha))(\sigma_{13}(\alpha)).\]
    Without loss of generality, we suppose 
    \[
    \mathfrak{p}\cap\Z[\zeta_{16}+\zeta_{16}^7]=(\sigma_1(\alpha)).
    \]
     Note $\Z[\zeta_{16}+\zeta_{16}^7]^{\times}=\langle1+\sqrt{2},-1\rangle$. Writing generators in $(\sigma_1(\alpha))$ as $\alpha(1+\sqrt{2})^n(-1)^k$ with $n,k\in\Z$, we have
    \begin{align*}
        &||\Sigma_{\Q(\zeta_{16}+\zeta_{16}^7)}(\alpha(-1)^k(1+\sqrt{2})^n)||^2\\
        =&2|\alpha(1+\sqrt{2})^n|^2+2|\sigma_3(\alpha(1+\sqrt{2})^n)|^2\\
        =&2(1+\sqrt{2})^{2n}|\alpha|^2+2(1-\sqrt{2})^{2n}|\sigma_3(\alpha)|^2\\
        =&2(1+\sqrt{2})^{2n}(a_p+b_p\sqrt{2})+2(1-\sqrt{2})^{2n}(a_p-b_p\sqrt{2}).
    \end{align*}
    Then, similarly with the proof of Theorem~\ref{theorem: p916}, we can conclude that $\alpha$ is the shortest vector in $\mathfrak{p}\cap\Z[\zeta_{16}+\zeta_{16}^7]$ with  length  $\sqrt{4a_p}$ and is also the shortest vector in $\mathfrak{p}$ with  length of \(\sqrt{2^{n}a_p}.\)

\end{proof}

Based on Algorithm~\ref{Algorithm-ap} and Theorem~\ref{theorem:p716}, we can obtain more efficient classical and quantum algorithms for computing $\lambda_1\left(\Sigma_{\Q(\zeta_{2^{n+1}})}(\mathfrak{p})\right)$. 

\noindent\textbf{The upper bound of the length of the shortest vector.} The analysis is similar with Subsection~\ref{subsection:upperBoundp916}. We conclude that
\[
\lambda_1\left(\Sigma_{\Q(\zeta_{2^{n+1}})}(\mathfrak{p})\right)< \sqrt[4]{2^{2n+1}p}.
\]
    The corresponding Minkowski bound is as follows:
    $$\lambda_1\left(\Sigma_{\Q(\zeta_{2^{n+1}})}(\mathfrak{p})\right)<\sqrt{2^{n}}\det(\mathfrak{p})^{\frac{1}{2^n}}=\sqrt{2^{n}}\left(N(\mathfrak{p})\sqrt{\Delta\Q({\zeta_{2^{n+1}}})}\right)^{\frac{1}{2^n}}
    =2^n\sqrt[4]{p}.$$
    Note that $\sqrt[4]{2^{2n+1}p}=2^{\frac{2n+1}{4}}\sqrt[4]{p}<2^n\sqrt[4]{p}$. So the upper bound we propose is tighter than the upper bound obtained from Minkowski's Theorem.
\section{Conclusions}
This paper presents a detailed analysis of optimization SVP for prime ideals in power-of-two cyclotomic fields $\Z[\zeta_{2^{n+1}}]$, focusing on primes $p\equiv7,9\pmod{16}$. By combining the principal ideal approach and the decomposition subfield approach, it is demonstrated that to find the shortest vector in an ideal lattice is reduced to find a generator of that ideal restricted to  specific subrings. 
specifically, we provide a precise characterization of the length of the shortest vector for prime ideals in $\Z[\zeta_{2^{n+1}}]$ lying over $p\equiv7,9\pmod{16}$, given by $\sqrt{2^{n}a_p}$.
We show a new and tighter upper bound of $\sqrt[4]{2^{2n+1}p}$ for this length, improving upon the bound derived from Minkowski's Theorem.

It is interesting to explore whether the generator-shortest vector correspondence holds in other number fields or for more general ideals. Extending this analysis to other classes of primes and different cyclotomic fields would also be a natural next step. 

\bibliographystyle{splncs04}
\bibliography{idealSVP}
\end{document}